\documentclass[conference,11pt,onecolumn]{IEEEtran}
\IEEEoverridecommandlockouts

\usepackage[utf8]{inputenc} 
\usepackage[T1]{fontenc}
\usepackage{url}
\usepackage{ifthen}
\usepackage{cite}
\usepackage[cmex10]{amsmath} 
\usepackage{amssymb}
\usepackage{amsthm}
\usepackage{cuted}
\usepackage{subfigure}
\usepackage{booktabs}
\usepackage{hhline}
\usepackage{array}
\usepackage{graphicx}
\usepackage{textcomp}
\usepackage{xcolor}
\usepackage{colortbl}

\newtheorem{lem}{Lemma}

\newtheorem{theorem}{Theorem}

\newtheorem{remark}{Remark}

\usepackage{bbm}
\usepackage{amsmath}
\usepackage{amssymb}
\usepackage{amsfonts}
\usepackage{epsfig}
\usepackage{calc}
\usepackage{color}
\usepackage[all]{xy}
\usepackage{bm}
\usepackage{enumerate}
\usepackage{pdfsync}
\usepackage{framed}
\usepackage{blkarray}
\usepackage{subcaption}
\usepackage{enumitem}
\usepackage{hyperref}
\usepackage[linesnumbered,ruled,vlined]{algorithm2e}
\usepackage{comment}
\usepackage{url}
\usepackage{multirow}

\usepackage{caption} 
\renewcommand{\arraystretch}{1.8}
\def\BibTeX{{\rm B\kern-.05em{\sc i\kern-.025em b}\kern-.08em
    T\kern-.1667em\lower.7ex\hbox{E}\kern-.125emX}}
\begin{document}

\title{Deterministic Distributed Linearly Separable Computation via Nested Cyclic Codes}
\author{
  \IEEEauthorblockN{Anu Mathew}
  \IEEEauthorblockA{IIT Palakkad, India  \\ anumundackal@gmail.com}
  
  \and 
  \IEEEauthorblockN{M. Nikhil Krishnan}
  \IEEEauthorblockA{IIT Palakkad, India \\
  nikhilkrishnan.m@gmail.com}
  \thanks{The work of M. Nikhil Krishnan is supported by the DST-INSPIRE Faculty Fellowship from the Department of Science and Technology (DST), and the ARG MATRICS grant from the Anusandhan National Research Foundation (ANRF).}}

\maketitle

\begin{abstract}
We study the $(K,N,S,M,K_c)$ distributed linearly separable computation (DLSC) setting, where a master employs $N$ workers with a per-worker computation load of $M$ to recover $K_c$ deterministic linear combinations of $K$ messages in the presence of up to $S$ stragglers. We propose a DLSC framework based on nested cyclic codes that achieves $M=\frac{K}{N}(S+K_c)$. For sufficiently large $K$, we establish the matching converse
$M\geq\frac{K}{N}(S+K_c)$, proving the asymptotic
computation load optimality of the proposed framework. By instantiating cyclic locally recoverable codes (LRCs) as base codes in the framework, the master can recover the desired $K_c$ linear combinations by contacting fewer than $N-S$ workers under many straggler scenarios. This essentially softens the straggler
threshold without increasing the computation load or communication
cost. Finally, we extend the construction to reduce the per-worker communication load by a factor of $m$, while achieving an asymptotically optimal computation load of $M=\frac{K}{N}(S+mK_c)$.
\end{abstract}

\section{Introduction}

Distributed computing enables large-scale computations across multiple
workers, but its performance can be affected by slow or failed workers,
known as \emph{stragglers}~\cite{Dean2013}. Coding-theoretic techniques
have been developed to provide straggler resilience through structured
redundancy~\cite{kiani2018exploitation,park2018hierarchical,
Li2016Unified,Li2018Fundamental}. Among these approaches, gradient coding (GC), introduced in~\cite{Tandon2017}, has been widely studied for straggler-resilient
distributed gradient computation. GC can be viewed as a special case
of DLSC with $K_c=1$, where the master seeks to recover the sum of the
partial gradients from any $N-S$ workers with a per-worker computation
load of $\frac{K}{N}(S+1)$.
Several alternative GC constructions have subsequently been developed,
including cyclic MDS-based schemes~\cite{Raviv2020},
Reed--Solomon-based constructions~\cite{Halbawi2018RS},
clustered gradient coding~\cite{Buyukates2023}, and schemes
addressing computation--communication
tradeoffs~\cite{Ye2018,Kadhe2020CommGC}.

Beyond gradient computation, coding-based approaches have also been extensively studied for distributed matrix--vector and matrix--matrix multiplication and
polynomial computations~\cite{dutta2016shortdot,Suh2017Matrix,
ramamoorthy2019universally,Das2019Convolutional,Ramamoorthy2020,
Lee2018,yu2017polynomial,Dutta2020Optimal,Yu2020,
NikhilKrishnan2021}. In a related direction, the DLSC problem was formulated
in~\cite{Wan2021} for recovering $K_c$ linear combinations of
$K$ messages from any $N-S$ workers.
The work \cite{Wan2021} considers randomly chosen demand coefficients with
high-probability recovery and the minimum computation load
$M=\frac{K}{N}(S+1)$, while characterizing the corresponding minimum
communication cost. 

In many practical applications, the desired linear combinations are task-dependent rather than randomly chosen. Consider a computing cluster evaluating multiple tasks over distributed data, such as hyperparameter tuning (requiring a distinct aggregate gradient for each hyperparameter), or large-scale matrix-vector and matrix-matrix multiplications involving specific matrices (e.g., the DFT matrix). Such applications motivate the deterministic recovery of arbitrary prescribed linear combinations under a fixed data assignment.

\subsection*{Our Contributions}
\begin{itemize}
\item We develop a DLSC framework using a \emph{nested family of cyclic codes}. This framework guarantees the \emph{deterministic recovery} of any $K_c$ prescribed linear combinations of $K$ messages from the responses of any $N-S$ workers, achieving a per-worker computation load of $M=\frac{K}{N}(S+K_c)$.

\item We establish a matching converse for sufficiently
large $K$, proving the asymptotic computation-load optimality
of our proposed scheme.

\item We instantiate cyclic LRCs to provide
a \emph{soft straggler threshold}, enabling recovery
beyond $S$ stragglers under certain straggler patterns
without compromising the computation or communication load.

\item We extend the construction to provide a flexible computation--communication tradeoff. By increasing the computation load to $M=\frac{K}{N}(S+mK_c)$, the per-worker communication load is reduced by a factor of $m$. We show this tradeoff remains asymptotically optimal.
\end{itemize}
\subsection{Notation}
For integers $a\leq b$, define
$[a:b]\triangleq\{a,a+1,\ldots,b\}$. Uppercase bold letters (e.g., $\mathbf{A}$) denote matrices, uppercase
calligraphic letters (e.g., $\mathcal{C}$ and $\mathcal{S}$) denote
codes and sets, and lowercase bold letters (e.g., $\mathbf{a}$) denote vectors. For a matrix $\mathbf A\in\mathbb F^{m\times n}$, we use zero-based
indexing, with row indices $i\in[0:m-1]$ and column indices
$j\in[0:n-1]$. The notations $\mathbf A(i,:)$,
$\mathbf A(:,j)$, and $a_{i,j}$ denote the row indexed by $i$, the
column indexed by $j$, and the $(i,j)$-th entry of $\mathbf A$,
respectively. For an index set $\mathcal S$, $\mathbf A(\mathcal S,:)$
denotes the submatrix of $\mathbf A$ formed by
the rows indexed by $\mathcal S$. For index sets $\mathcal R$ and $\mathcal J$,
$\mathbf A(\mathcal R,\mathcal J)$ denotes the
submatrix formed by the rows indexed by $\mathcal R$
and columns indexed by $\mathcal J$. The row and column spaces of $\mathbf A$ are
denoted by $\operatorname{Row}(\mathbf A)$
and $\operatorname{Col}(\mathbf A)$, respectively. We use $\mathbb{F}$ to denote a field and $\mathbb{F}_q$ to denote the
finite field with $q$ elements. The $m\times m$ identity matrix is
denoted by $\mathbf{I}_m$. For any vector $\mathbf{a}\in\mathbb{F}^N$,
its support and Hamming weight are defined as
$\operatorname{supp}(\mathbf{a})
=\{i\in[0:N-1]\mid a_i\neq0\}$ and
$\operatorname{wt}(\mathbf{a})
=|\operatorname{supp}(\mathbf{a})|$,
respectively.

\subsection{System Model and Problem Formulation}
\label{sys-model}
\begingroup
\setlength{\abovedisplayskip}{4pt}
\setlength{\belowdisplayskip}{4pt}
\setlength{\abovedisplayshortskip}{3pt}
\setlength{\belowdisplayshortskip}{3pt}

We consider a $(K,N,S,M,K_c)$ distributed linearly separable
computation (DLSC) system with a master and $N$ workers
$W_0,\ldots,W_{N-1}$. A dataset is partitioned into $K$ data fragments
$D_0,\ldots,D_{K-1}$, where $K$ is assumed to be a multiple of $N$.
For each data fragment $D_k$, let the column vector $P_k\triangleq f_k(D_k)\in\mathbb{F}^{L\times 1}, k\in[0:K-1],$
denote the corresponding message of $L$ symbols. The messages are assumed to be mutually independent.

The $K$ data fragments are assigned to the $N$ workers in an uncoded manner. Let $\mathcal{S}_i\subseteq[0:K-1]$ denote the indices of the data fragments assigned to worker $W_i$, such that $D_k$ is available at worker $W_i$ if and only if $k\in\mathcal{S}_i$. We consider a uniform data assignment with $|\mathcal{S}_i|=M$ for all $i\in[0:N-1]$, where $M$ denotes the per-worker computation load. Consistent with the literature, we adopt a cyclic data assignment in which a base assignment pattern is cyclically shifted across the workers and repeated over successive blocks of $N$ data fragments. For $K=\mu N$, where $\mu$ is a positive integer and $i\in[0:N-1]$, the index set of the data fragments assigned to worker $W_i$ is given by
\begin{equation*}
\mathcal{S}_i=
\bigcup_{t=0}^{\mu-1}
\left\{
tN+((k-i)\bmod N) \mid k\in\mathcal{S}_0^{(0)}
\right\},
\end{equation*}
where $\mathcal{S}_0^{(0)}\subseteq[0:N-1]$ denotes the indices of the data fragments assigned to worker $W_0$ in the first block, with $|\mathcal{S}_0^{(0)}|=M/\mu$. Owing to this structure, we refer to the proposed framework as \emph{cyclic DLSC}.

Each worker $W_i$ then computes
$P_k=f_k(D_k)$ for every $k\in\mathcal{S}_i$.
The master wishes to recover $K_c$ prescribed linear combinations of
the $K$ messages. Let $T$ denote the desired linear mapping specified
by a full-row-rank demand matrix
$\mathbf{A}\in\mathbb{F}^{K_c\times K}$. i.e., $T(P_0,\ldots,P_{K-1})
=
\mathbf{A}
\begingroup
\setlength{\arraycolsep}{2pt}
\renewcommand{\arraystretch}{0.65}
\begin{bmatrix}
P_0 & \cdots & P_{K-1}
\end{bmatrix}^{T} \endgroup$, where $\operatorname{rank}(\mathbf{A})=K_c$.

Note that the data assignment depends only on the number of desired linear combinations $K_c$ and can be performed
without prior knowledge of the demand matrix $\mathbf A$.
Thus, the same data placement supports different sets of
$K_c$ linear combinations without requiring data
redistribution among the workers.

Let $\mathbf{\Phi}=[\phi_{i,k}]\in\mathbb{F}^{N\times K}$ denote the
encoding matrix corresponding to the prescribed demand matrix
$\mathbf A$, with
$\phi_{i,k}=0$ for all $k\notin\mathcal S_i$. Worker $W_i$ transmits the coded message
\begin{equation*}
Y_i
=
\sum_{k\in\mathcal{S}_i}\phi_{i,k}P_k,
\qquad i\in[0:N-1].
\end{equation*}
The master must recover the desired $K_c$ linear combinations from
the responses of any $N-S$ workers. Hence, for every response set $\mathcal{R}\subseteq[0:N-1]$ with $|\mathcal{R}|=N-S$, the corresponding $(N-S)\times K$
submatrix $\mathbf{\Phi}_{\mathcal{R}}$ must satisfy $\operatorname{Row}(\mathbf{A})
\subseteq
\operatorname{Row}(\mathbf{\Phi}_{\mathcal{R}}).$
Note that this necessitates $K_c\leq N-S.$
\endgroup

\section{Proposed Cyclic DLSC Framework}
\label{sec:cyclic-dlsc-framework}
\begingroup
\setlength{\abovedisplayskip}{4pt}
\setlength{\belowdisplayskip}{4pt}
\setlength{\abovedisplayshortskip}{3pt}
\setlength{\belowdisplayshortskip}{3pt}
In this section, we present a cyclic $(K,N,S,M,K_c)$ DLSC scheme based on a nested
family of cyclic codes. The support structure of the encoding matrix $\mathbf{\Phi}$ is defined by minimum-weight codewords selected from these codes, while the exact encoding coefficients are determined by the prescribed demand matrix $\mathbf A$. The resulting scheme achieves a per-worker computation load $M=\frac{K}{N}(S+K_c)$.

For ease of presentation, we first consider $K=N$ and subsequently
extend the construction to $K$ being a positive integer multiple of
$N$. The proposed framework can be realized using suitable families
of nested cyclic codes, including Reed--Solomon (RS) codes. Moreover, choosing LRCs enables the framework to exploit local repair
and recover from more than $S$ stragglers under favorable patterns.
Thus, the LRC-based construction relaxes the hard recovery threshold of $N-S$. We first present the general construction, followed by specific realizations using RS codes and LRCs.

\subsection{General Construction Based on Nested Cyclic Codes}
\label{subsubsec:general-nested-construction}

Consider a family of $K_c$ nested length-$N$ cyclic codes $\mathcal{C}_1\subsetneq\mathcal{C}_2\subsetneq\cdots
\subsetneq\mathcal{C}_{K_c},$
where $\mathcal{C}_{K_c}$ is the supercode. Let $k_i$ and $d_i$
denote the dimension and minimum distance of $\mathcal{C}_i$,
respectively. For each $\mathcal C_i$, select a minimum-distance
codeword $\begingroup
\setlength{\arraycolsep}{2pt}
\renewcommand{\arraystretch}{0.82}
\mathbf{c}_i
=
\begin{bmatrix}
c_{i,0} & c_{i,1} & \cdots & c_{i,N-1}
\end{bmatrix}, i \in [1:K_c].
\endgroup$
The nested codes and the selected minimum-distance codewords are
required to satisfy
\begin{equation}
\begin{aligned}
d_i &= S+K_c-i+1,\\
k_i &= k_1+i-1, \qquad i \in [1:K_c],\\
\operatorname{supp}(\mathbf c_{K_c})
&\subsetneq \operatorname{supp}(\mathbf c_{K_c-1})
\subsetneq \cdots
\subsetneq \operatorname{supp}(\mathbf c_1),
\end{aligned}
\label{eq:nested-constraints}
\end{equation}
with $\operatorname{wt}(\mathbf c_i)=d_i$.
Consequently, the supercode $\mathcal C_{K_c}$ has minimum distance
$d_{K_c}=S+1$, while
$\operatorname{wt}(\mathbf c_1)=d_1=S+K_c$.
The existence of minimum-distance codewords satisfying
\eqref{eq:nested-constraints} depends on the underlying code
family and is established for the RS and LRC constructions
presented subsequently.

Define the block matrix 
\[
\begingroup
\setlength{\arraycolsep}{2pt}
\renewcommand{\arraystretch}{0.75}
\mathbf{\Delta}\triangleq
\left[
\mathbf{\Delta}_0\;
\mathbf{\Delta}_1\;
\cdots\;
\mathbf{\Delta}_{N-1}
\right]
\in\mathbb{F}^{N\times NK_c},
\endgroup
\]
where the first block is formed from the selected minimum-weight
codewords as $\begingroup
\setlength{\arraycolsep}{2pt}
\renewcommand{\arraystretch}{0.75}
\mathbf{\Delta}_0
=
\begin{bmatrix}
\mathbf{c}_1^T & \mathbf{c}_2^T & \cdots & \mathbf{c}_{K_c}^T
\end{bmatrix}.
\endgroup$
The remaining blocks
$\mathbf{\Delta}_1,\ldots,\mathbf{\Delta}_{N-1}$ are obtained by
successive cyclic shifts of the columns of $\mathbf{\Delta}_0$.
i.e., for $j \in [0:N-1]$,
\[
\begingroup
\setlength{\arraycolsep}{2pt}
\renewcommand{\arraystretch}{0.65}
\mathbf{\Delta}_j\triangleq
\begin{bmatrix}
c_{1,j} & c_{2,j} & \cdots & c_{K_c,j}\\
c_{1,j+1} & c_{2,j+1} & \cdots & c_{K_c,j+1}\\
\vdots & \vdots & \ddots & \vdots\\
c_{1,j+N-1} & c_{2,j+N-1} & \cdots & c_{K_c,j+N-1}
\end{bmatrix},
\endgroup
\]
with all indices taken modulo $N$. 
The block $\mathbf{\Delta}_j$ corresponds to data fragment $D_j$, and
the data assignment is specified by $j\in\mathcal{S}_i
\Longleftrightarrow
\mathbf{\Delta}_j(i,:)\neq\mathbf{0}, i,j \in [0:N-1].$

By the nested-support property, the nonzero row support of each block
$\mathbf{\Delta}_j$ is determined by the corresponding cyclic shift of
$\operatorname{supp}(\mathbf{c}_1)$. Since
$\operatorname{wt}(\mathbf{c}_1)=S+K_c$, each block has exactly
$S+K_c$ nonzero rows. Moreover, since the blocks are successive cyclic
shifts, each worker is assigned exactly $S+K_c$ data fragments.
Therefore, for $K=N$, the per-worker computation load is $M=S+K_c.$

Let $\mathbf{G}\in\mathbb{F}^{k_{K_c}\times N}$ be a generator matrix for the supercode $\mathcal{C}_{K_c}$, with rows $\mathbf{g}_e = [1\ \alpha^{s_e}\ \cdots\ \alpha^{(N-1)s_e}]$ for $e \in [0:k_{K_c}-1]$, where $\alpha$ is a primitive $N$-th root of unity. The distinct exponents $s_e$ are ordered such that the top $k_i$ rows of $\mathbf{G}$ form a generator matrix for $\mathcal{C}_i$, for each $i\in[1:K_c]$. Since all selected codewords belong to $\mathcal{C}_{K_c}$, we can factorize $\mathbf{\Delta}$ as
\[
\begingroup
\setlength{\arraycolsep}{2pt}
\renewcommand{\arraystretch}{0.75}
\mathbf{\Delta}
=
\mathbf{G}^{T}\mathbf{B}',
\qquad
\mathbf{B}'
=
\left[
\mathbf{B}'_0\;
\mathbf{B}'_1\;
\cdots\;
\mathbf{B}'_{N-1}
\right].
\endgroup
\]
where
$\mathbf{B}'\in\mathbb{F}^{k_{K_c}\times NK_c}$ and each
$\mathbf{B}_j'\in\mathbb{F}^{k_{K_c}\times K_c}$ contains the message
vectors corresponding to the codewords in $\mathbf{\Delta}_j$.
For a selected minimum-distance codeword $\mathbf{c}_1 \in \mathcal{C}_1$, the top $k_1$ rows of $\mathbf{G}$ can always be ordered such that the last message coefficient (associated with row $k_1-1$) is non-zero. This arrangement guarantees that the bottom $K_c \times K_c$ submatrix of $\mathbf{B}'_j$ is nonsingular for every $j\in[0:N-1]$ (see Appendix~\ref{app:Bj-nonsingular}).

Let $\mathbf A=[a_{i,j}]\in\mathbb F^{K_c\times N}$ be the prescribed
full-row-rank demand matrix, and let $\mathbf A(:,j)$ denote the $j$-th column. Since the bottom $K_c\times K_c$ submatrix of
$\mathbf{B}_j'$ is nonsingular, there exists a unique vector
$\mathbf{v}_j$ in its column space whose last $K_c$ entries equal
$\mathbf A(:,j)$, i.e.,
\begin{equation*}
\begingroup
\setlength{\arraycolsep}{3pt}
\renewcommand{\arraystretch}{0.82}
\mathbf{v}_j=
\begin{bmatrix}
\star & \cdots & \star &
a_{0,j} &
a_{1,j} &
\cdots &
a_{K_c-1,j}
\end{bmatrix}^{T},
j \in [0:N-1],
\endgroup
\end{equation*}
where ``$\star$'' denotes entries determined by the corresponding
linear combination of the columns of $\mathbf{B}_j'$.

We define $\begingroup
\setlength{\arraycolsep}{3pt}
\renewcommand{\arraystretch}{0.62}
\mathbf{B}
\triangleq
\begin{bmatrix}
\mathbf{v}_0 &
\mathbf{v}_1 &
\cdots &
\mathbf{v}_{N-1}
\end{bmatrix}.
\endgroup$
By construction, the last $K_c$ rows of $\mathbf B$ coincide with the
demand matrix $\mathbf A$. The encoding matrix is then defined as $\begingroup
\setlength{\arraycolsep}{3pt}
\renewcommand{\arraystretch}{0.82}
\mathbf{\Phi}
\triangleq
\mathbf{G}^{T}\mathbf{B}.
\endgroup$

Since $\mathbf v_j\in\operatorname{Col}(\mathbf B_j')$,
$\mathbf G^T\mathbf v_j$ is a linear combination
of the columns of $\mathbf\Delta_j$.
Therefore, every zero row of $\mathbf\Delta_j$
corresponds to a zero entry in $\mathbf G^T\mathbf v_j$.
Hence, $\operatorname{supp}\!\left(\mathbf\Phi(:,j)\right)
\subseteq
\{i\in[0:N-1]:\mathbf\Delta_j(i,:)\neq\mathbf 0\},$
and the encoding matrix $\mathbf\Phi$ is consistent with respect to the data assignment $\{\mathcal{S}_i\}$.

The coded computations transmitted by the $N$ workers can be
collectively expressed as 
$$\begingroup
\setlength{\arraycolsep}{3pt}
\renewcommand{\arraystretch}{0.82}
\mathbf{Y}
=
\mathbf{G}^{T}\mathbf{B}
\begin{bmatrix}
P_0 & P_1 & \cdots & P_{N-1}
\end{bmatrix}^{T}.
\endgroup$$
Here, the row indexed by $i$ of $\mathbf Y$ represents the coded computation
transmitted by worker $W_i$, $i\in[0:N-1]$.

Let $\mathcal R\subseteq[0:N-1]$, with $|\mathcal R|=N-S$, denote an
arbitrary set of nonstraggling workers. The master receives $\begingroup
\setlength{\arraycolsep}{2pt}
\renewcommand{\arraystretch}{0.75}
\mathbf Y_{\mathcal R}
=
\mathbf G^T(\mathcal R,:)\mathbf B
\begin{bmatrix}
P_0 & P_1 & \cdots & P_{N-1}
\end{bmatrix}^{T}.
\endgroup$
Each column of $\begingroup
\setlength{\arraycolsep}{2pt}
\renewcommand{\arraystretch}{0.75}
\mathbf Y
=
\mathbf G^T\mathbf B
\begin{bmatrix}
P_0 & P_1 & \cdots & P_{N-1}
\end{bmatrix}^{T}
\endgroup$
is a codeword of $\mathcal C_{K_c}$. Since
$d_{K_c}=S+1$, erasure decoding from the $N-S$ available coordinates
recovers the corresponding message vectors, yielding $\begingroup
\setlength{\arraycolsep}{2pt}
\renewcommand{\arraystretch}{0.75}
\mathbf B
\begin{bmatrix}
P_0 & P_1 & \cdots & P_{N-1}
\end{bmatrix}^{T}.
\endgroup$
Since the last $K_c$ rows of $\mathbf B$ coincide with $\mathbf A$,
the master obtains $\begingroup
\setlength{\arraycolsep}{2pt}
\renewcommand{\arraystretch}{0.75}
\mathbf A
\begin{bmatrix}
P_0 & P_1 & \cdots & P_{N-1}
\end{bmatrix}^{T},
\endgroup$
which gives the $K_c$ desired linear combinations. Thus we have the following theorem.

\begin{theorem}
\label{thm:general-cyclic-dlsc}
Let
$\mathcal{C}_1\subsetneq\mathcal{C}_2\subsetneq\cdots
\subsetneq\mathcal{C}_{K_c}$ be nested
length-$N$ cyclic codes with minimum distances
$d_i=S+K_c-i+1$, $i \in [1:K_c]$, and let their selected
minimum-distance codewords have nested supports. Then, for $K=N$ and
any full-row-rank demand matrix
$\mathbf{A}\in\mathbb{F}^{K_c\times N}$, the proposed construction
yields a cyclic $(K,N,S,M,K_c)$ DLSC scheme with $M=S+K_c$
that guarantees deterministic recovery of any $K_c$ desired linear
combinations from the responses of any $N-S$ workers.
\end{theorem}

For $K=\mu N$, where $\mu$ is a positive integer,
the $K$ data fragments are divided into $\mu$
groups of size $N$, with the cyclic data assignment
applied independently to each group.
Each worker combines its encoded computations
across the $\mu$ groups into a single response
before transmission. Consequently, the per-worker
computation load scales to
$M=\mu(S+K_c)=\frac{K}{N}(S+K_c)$,
while deterministic recovery from any $N-S$
workers is preserved.

\begin{theorem}
\label{thm:converse}
Suppose $N-S\geq K_c$ and
$K>(K_c-1)\binom{N}{N-S-K_c+1}$.
Any cyclic $(K,N,S,M,K_c)$ DLSC scheme guaranteeing deterministic
recovery of any $K_c$ desired linear combinations
from any $N-S$ workers must satisfy $M\geq\frac{K}{N}(S+K_c).$
\end{theorem}
The proof is given in Appendix~\ref{app:converse}. Together with Theorem~\ref{thm:general-cyclic-dlsc} and its extension to $K=\mu N$,
this establishes the asymptotic computation-load optimality of the proposed
construction.
\subsection{RS- and LRC-Based Cyclic DLSC Constructions}
\label{subsec:rs-lrc}
We next illustrate the proposed cyclic DLSC framework using Reed--Solomon (RS) codes and cyclic Tamo--Barg LRCs~\cite{Tamo2016}. 
For ease of illustration, we consider specific parameters,
while the construction naturally extends to more general parameters. We consider
$N=K=12$, $S=5$, and $K_c=3$ over $\mathbb F_{13}$,
with the desired linear combinations specified by
the full-row-rank demand matrix:
\begin{equation}
\begingroup
\setlength{\arraycolsep}{2pt}
\renewcommand{\arraystretch}{0.72}
\mathbf A=
\begin{bmatrix}
1&1&1&1&1&1&1&1&1&1&1&1\\
1&2&3&4&5&6&7&8&9&10&11&12\\
1&4&9&3&12&10&10&12&3&9&4&1
\end{bmatrix}.
\endgroup
\label{eq:common-demand}
\end{equation}

\subsubsection{Reed--Solomon-Based Construction}
\label{subsubsec:rs-construction}

For the RS realization, let $\alpha=2$ be a primitive $12$-th root of unity in $\mathbb{F}_{13}$. We consider a nested sequence of RS codes $\mathcal{C}_1 \subsetneq \mathcal{C}_2 \subsetneq \mathcal{C}_3$ with respective $[n, k_i, d_i]$ parameters $[12,5,8]$, $[12,6,7]$, and $[12,7,6]$. Each code $\mathcal{C}_i$ is constructed by evaluating all polynomials of degree at most $k_i-1$ at the points $1,\alpha,\ldots,\alpha^{11}$. The generator matrix, whose first $k_i$ rows generate
$\mathcal C_i$, $i=1,2,3$, is given by:
\[
\begingroup
\setlength{\arraycolsep}{2pt}
\renewcommand{\arraystretch}{0.72}
\mathbf G_{\rm RS}
=
\begin{bmatrix}
1&1&1&\cdots&1\\
1&2&2^2&\cdots&2^{11}\\
1&2^2&(2^2)^2&\cdots&(2^2)^{11}\\
\vdots&\vdots&\vdots&\ddots&\vdots\\
1&2^6&(2^6)^2&\cdots&(2^6)^{11}
\end{bmatrix}
\in\mathbb F_{13}^{7\times12}.
\endgroup
\]
To construct minimum-distance codewords with nested supports,
choose the root sets $\mathcal T_1=\{\alpha^0,\alpha^1,\alpha^2,\alpha^3\},
\mathcal T_2=\mathcal T_1\cup\{\alpha^4\},
\mathcal T_3=\mathcal T_2\cup\{\alpha^5\},$
and define the corresponding encoding polynomials as:
\[
u_i(x)=\prod_{\beta\in\mathcal T_i}(x-\beta),
\qquad i=1,2,3.
\]
Evaluating $u_i(x)$ at $1,\alpha,\ldots,\alpha^{11}$ gives
\begin{equation*}
\begin{aligned}
\mathbf c_1&=[0,0,0,0,10,11,10,5,7,3,6,1],\\
\mathbf c_2&=[0,0,0,0,0,7,12,1,3,6,3,4],\\
\mathbf c_3&=[0,0,0,0,0,0,7,5,9,7,12,4].
\end{aligned}
\label{eq:rs-min-codewords}
\end{equation*}
Hence, $\operatorname{wt}(\mathbf c_i)=d_i$ and
$\operatorname{supp}(\mathbf c_3)\subsetneq
\operatorname{supp}(\mathbf c_2)\subsetneq
\operatorname{supp}(\mathbf c_1)$.
Therefore, the nested RS codes and the selected minimum-distance
codewords satisfy \eqref{eq:nested-constraints}.
The first block of $\mathbf \Delta$ is $\begingroup\setlength{\arraycolsep}{2pt}\renewcommand{\arraystretch}{0.72}\mathbf\Delta_0=\begin{bmatrix}\mathbf c_1^T&\mathbf c_2^T&\mathbf c_3^T\end{bmatrix},\endgroup$ and the remaining blocks $\mathbf\Delta_1,\ldots,\mathbf\Delta_{11}$ are obtained by successively cyclically shifting the rows of $\mathbf\Delta_0$ upward by one position at a time. In particular, $\operatorname{supp}(\mathbf c_1)$ and its
successive cyclic shifts specify the data placement. Since
$\operatorname{wt}(\mathbf c_1)=8$, each data partition is assigned
to eight workers and, by cyclic symmetry, each worker stores eight
partitions; hence, $M=8=S+K_c$.

With $\mathbf\Delta_0=\mathbf G_{\rm RS}^{T}\mathbf{B}_0'$, the corresponding
message matrix is
\[
\begingroup
\setlength{\arraycolsep}{2pt}
\renewcommand{\arraystretch}{0.72}
\mathbf{B}_0'=
\begin{bmatrix}
\mathbf b_1&\mathbf b_2&\mathbf b_3
\end{bmatrix}
=
\begin{bmatrix}
12&3&8\\
10&8&7\\
5&8&12\\
11&11&7\\
1&8&2\\
0&1&2\\
0&0&1
\end{bmatrix}.
\endgroup
\]
The bottom $K_c \times K_c$ submatrix of $\mathbf{B}_0'$ is nonsingular. The same property holds for the matrices $\mathbf{B}_j'$ associated with the other cyclically shifted blocks $\mathbf\Delta_j$. Therefore, following the general construction, the columns of the demand matrix in \eqref{eq:common-demand} can be embedded in the last $K_c$ entries of the corresponding vectors $\mathbf v_j\in\operatorname{Col}(\mathbf{B}_j')$. Defining
$\mathbf B_{\rm RS}=[\mathbf v_0\,\cdots\,\mathbf v_{11}]$ and
$\mathbf\Phi_{\rm RS}=\mathbf G_{\rm RS}^T\mathbf B_{\rm RS}$,
the last $K_c$ rows of $\mathbf B_{\rm RS}$ equal $\mathbf A$.
Since $\mathcal C_3$ has $d_3=6=S+1$, the master can recover the
desired three linear combinations from any $N-S=7$ workers.

\subsubsection{Cyclic Tamo--Barg LRC-Based Construction}
\label{subsubsec:lrc-construction}

We next instantiate the proposed framework using cyclic Tamo--Barg LRCs~\cite{Tamo2016}. A linear code of length $N$ has disjoint $(r,\delta)$-locality if its symbols can be partitioned into disjoint \emph{local repair groups} of size $r+\delta-1$, such that any $\delta-1$ erasures within a group are recoverable from the remaining $r$ symbols. Cyclic Tamo--Barg codes naturally possess this property. Following~\cite{Tamo2016}, their roots are specified by an index set $\mathcal U=\mathcal L\cup\mathcal D$, whose elements dictate the powers of a primitive $N$-th root of unity. The {\it locality index set}
\[
\mathcal L=
\left\{
j\in[0:N-1]:
j\bmod(r+\delta-1)\in[0:\delta-2]
\right\}
\]
establishes this $(r,\delta)$-locality structure, while the {\it distance index set} $\mathcal D$ comprises $d_{\min}-1$ consecutive integers (modulo $N$) to guarantee a minimum distance of at least $d_{\min}$ via the BCH bound. The resulting code is minimum-distance optimal, as it achieves the upper bound in \cite[Theorem~2.1]{Kamath2014} with equality.

To construct the required nested family, we fix a common locality index set $\mathcal L$ and select nested distance index sets $\mathcal D_{K_c}\subsetneq\cdots\subsetneq\mathcal D_1$. The corresponding root index sets $\mathcal U_i=\mathcal L\cup\mathcal D_i$ yield a nested sequence of cyclic LRCs $\mathcal C_1\subsetneq\cdots\subsetneq \mathcal C_{K_c}$. Crucially, maintaining a common $\mathcal L$ across all root sets ensures that every code in the nested sequence possesses the exact same local repair groups. The distance index sets are chosen to guarantee $d_i=S+K_c-i+1$.

In order to satisfy the minimum-distance conditions
in \eqref{eq:nested-constraints}, the dimension of
the smallest subcode $\mathcal C_1$ must be of the
form $k_1=ar+x$, where $a\in\mathbb Z_{\geq 0}$,
$x\in\mathbb Z$, $r\geq K_c$, and
$1\leq x\leq r-K_c+1$.

To illustrate, let $N=K=12$, $S=5$, $K_c=3$, and $(r,\delta)=(3,2)$. We select the common locality index set $\mathcal L=\{0,4,8\}$ along with the nested distance index sets $\mathcal D_1=\{0,1,\dots,6\}$, $\mathcal D_2=\{0,1,\dots,5\}$, and $\mathcal D_3=\{0,1,\dots,4\}$. This configuration yields the nested sequence of cyclic LRCs $\mathcal C_1\colon [12,4,8] \subsetneq \mathcal C_2\colon [12,5,7] \subsetneq \mathcal C_3\colon [12,6,6]$. The corresponding generator matrix $\mathbf G_{\rm LRC}$,
whose first $k_i$ rows generate $\mathcal C_i$, is

\begin{equation*}
\resizebox{0.75\textwidth}{!}{%
$
\begingroup
\setlength{\arraycolsep}{0.8pt}
\renewcommand{\arraystretch}{0.72}
\mathbf G_{\rm LRC}=
\begin{bmatrix}
    \cellcolor{red!12}1&
\cellcolor{blue!12}2&
\cellcolor{green!12}2^2&
\cellcolor{red!12}2^3&
\cellcolor{blue!12}2^4&
\cellcolor{green!12}2^5&
\cellcolor{red!12}2^6&
\cellcolor{blue!12}2^7&
\cellcolor{green!12}2^8&
\cellcolor{red!12}2^9&
\cellcolor{blue!12}2^{10}&
\cellcolor{green!12}2^{11}\\

\cellcolor{red!12}1&
\cellcolor{blue!12}2^2&
\cellcolor{green!12}{(2^2)}^2&
\cellcolor{red!12}{(2^2)}^3&
\cellcolor{blue!12}{(2^2)}^4&
\cellcolor{green!12}{(2^2)}^5&
\cellcolor{red!12}{(2^2)}^6&
\cellcolor{blue!12}{(2^2)}^7&
\cellcolor{green!12}{(2^2)}^8&
\cellcolor{red!12}{(2^2)}^9&
\cellcolor{blue!12}{(2^2)}^{\!10}&
\cellcolor{green!12}{(2^2)}^{\!11}\\

\cellcolor{red!12}1&
\cellcolor{blue!12}2^3&
\cellcolor{green!12}{(2^3)}^2&
\cellcolor{red!12}{(2^3)}^3&
\cellcolor{blue!12}{(2^3)}^4&
\cellcolor{green!12}{(2^3)}^5&
\cellcolor{red!12}{(2^3)}^6&
\cellcolor{blue!12}{(2^3)}^7&
\cellcolor{green!12}{(2^3)}^8&
\cellcolor{red!12}{(2^3)}^9&
\cellcolor{blue!12}{(2^3)}^{\!10}&
\cellcolor{green!12}{(2^3)}^{\!11}\\

\cellcolor{red!12}1&
\cellcolor{blue!12}2^5&
\cellcolor{green!12}{(2^5)}^2&
\cellcolor{red!12}{(2^5)}^3&
\cellcolor{blue!12}{(2^5)}^4&
\cellcolor{green!12}{(2^5)}^5&
\cellcolor{red!12}{(2^5)}^6&
\cellcolor{blue!12}{(2^5)}^7&
\cellcolor{green!12}{(2^5)}^8&
\cellcolor{red!12}{(2^5)}^9&
\cellcolor{blue!12}{(2^5)}^{\!10}&
\cellcolor{green!12}{(2^5)}^{\!11}\\

\cellcolor{red!12}1&
\cellcolor{blue!12}2^6&
\cellcolor{green!12}{(2^6)}^2&
\cellcolor{red!12}{(2^6)}^3&
\cellcolor{blue!12}{(2^6)}^4&
\cellcolor{green!12}{(2^6)}^5&
\cellcolor{red!12}{(2^6)}^6&
\cellcolor{blue!12}{(2^6)}^7&
\cellcolor{green!12}{(2^6)}^8&
\cellcolor{red!12}{(2^6)}^9&
\cellcolor{blue!12}{(2^6)}^{\!10}&
\cellcolor{green!12}{(2^6)}^{\!11}\\

\cellcolor{red!12}1&
\cellcolor{blue!12}2^7&
\cellcolor{green!12}{(2^7)}^2&
\cellcolor{red!12}{(2^7)}^3&
\cellcolor{blue!12}{(2^7)}^4&
\cellcolor{green!12}{(2^7)}^5&
\cellcolor{red!12}{(2^7)}^6&
\cellcolor{blue!12}{(2^7)}^7&
\cellcolor{green!12}{(2^7)}^8&
\cellcolor{red!12}{(2^7)}^9&
\cellcolor{blue!12}{(2^7)}^{\!10}&
\cellcolor{green!12}{(2^7)}^{\!11}
\end{bmatrix}
\endgroup
$
}
\end{equation*}

The three colors distinguish the local repair groups: $\mathcal G_1=\{0,3,6,9\},\mathcal G_2=\{1,4,7,10\},\mathcal G_3=\{2,5,8,11\}$. Each group forms a local $[4,3,2]$ MDS code, such that any three code symbols within a group uniquely determine the remaining symbol. The specific minimum-weight codewords satisfying \eqref{eq:nested-constraints} are obtained as:
\begin{equation*}
\begin{aligned}
\mathbf c_1&=[0,1,8,0,8,12,0,12,5,0,5,1],\\
\mathbf c_2&=[0,1,5,0,0,3,0,8,9,0,4,9],\\
\mathbf c_3&=[0,1,1,0,0,6,0,0,2,0,12,4],
\end{aligned}
\end{equation*}
with corresponding message vectors
$\mathbf m_1=[3,0,0,10]$,
$\mathbf m_2=[9,10,0,4,3]$, and
$\mathbf m_3=[11,4,9,2,9,4]$.
The procedure to obtain these specific vectors is detailed in Appendix~\ref{app:lrc-minweight}.

As in the RS-based construction, the first block of $\mathbf\Delta$ is
$\begingroup
\setlength{\arraycolsep}{2pt}
\renewcommand{\arraystretch}{0.72}
\mathbf\Delta_0=
\begin{bmatrix}
\mathbf c_1^T&\mathbf c_2^T&\mathbf c_3^T
\end{bmatrix},
\endgroup$
and the remaining blocks $\mathbf\Delta_1,\ldots,\mathbf\Delta_{11}$ are obtained by successively cyclically shifting the rows of $\mathbf\Delta_0$. Since $\operatorname{wt}(\mathbf c_1)=8$, the resulting cyclic data assignment achieves $M=8=S+K_c$.

Since $\mathbf\Delta_0=\mathbf G_{\rm LRC}^{T}\mathbf{B}_0'$,
the corresponding message matrix $\mathbf{B}_0'$, formed by
zero-padding $\mathbf m_1,\mathbf m_2$, and $\mathbf m_3$
to length six, is given by:
\[
\begingroup
\setlength{\arraycolsep}{2pt}
\renewcommand{\arraystretch}{0.72}
\mathbf{B}_0'
=
\begin{bmatrix}
3&9&11\\
0&10&4\\
0&0&9\\
10&4&2\\
0&3&9\\
0&0&4
\end{bmatrix}.
\endgroup
\]
The bottom $K_c\times K_c$ submatrix of $\mathbf{B}_0'$ is nonsingular. The same property holds for the matrices $\mathbf{B}_j'$ associated with the other cyclically shifted blocks $\mathbf\Delta_j$. Hence, following the identical procedure used in the RS-based construction, the columns of the demand matrix in \eqref{eq:common-demand} are embedded into the last $K_c$ entries of the corresponding vectors $\mathbf v_j\in\operatorname{Col}(\mathbf{B}_j')$. Thus,
we obtain
$\mathbf B_{\rm LRC}=[\mathbf v_0\,\cdots\,\mathbf v_{11}]$
and
$\mathbf\Phi_{\rm LRC}=\mathbf G_{\rm LRC}^T\mathbf B_{\rm LRC}$,
where the last $K_c$ rows of $\mathbf B_{\rm LRC}$
coincide with $\mathbf A$. Since $d_3=6=S+1$, the master successfully recovers the $K_c=3$ desired linear combinations from any $N-S=7$ workers.

In addition to this worst-case guarantee, the inherent locality structure allows recovery from specific straggler patterns containing more than $S$ stragglers. For example, consider a scenario with six straggling workers $\mathcal F=\{W_0,W_1,W_2,W_3,W_6,W_9\}.$. Although $|\mathcal F|=6>S$, local groups $\mathcal{G}_2$ and $\mathcal{G}_3$ (blue and green in $\mathbf G_{\rm LRC}$) contain only one straggler each ($W_1$ and $W_2$). Because each group forms a local $[4,3,2]$ MDS code, the responses of $W_1$ and $W_2$ are directly reconstructed from the surviving workers in their respective groups. This local repair reduces the number of unresolved stragglers to four, which is comfortably handled by the standard $S=5$ tolerance of $\mathcal{C}_3$. Thus, unlike the RS realization—which operates under the exact same $M=8$ load and worst-case $S=5$ threshold—the LRC approach exploits local repair to circumvent the hard straggler bottleneck without incurring any additional computation or communication overhead.

\subsection{Computation-Latency Analysis}
\label{subsec:latency-analysis}
We compare the computation latency of the RS- and LRC-based constructions under worker straggling. Following~\cite{Lee2018}, the computation time $T_i$ of worker $W_i$ is modeled as an independent shifted-exponential random variable,
\begin{equation}
\Pr(T_i\leq t)=
\begin{cases}
1-\exp[-\lambda(t-\tau)], & t\geq\tau,\\
0, & t<\tau,
\end{cases}
\label{eq:worker-latency-model}
\end{equation}
where $\tau$ represents the deterministic minimum computation time, and the exponential tail (parameterized by $\lambda$) models the unpredictable delay introduced by straggling workers.

Let $T_{(0)}\leq T_{(1)}\leq\cdots\leq T_{(N-1)}$ denote the ordered worker completion times. The RS construction strictly requires $N-S$ responses, making its completion time $T_{\mathrm{RS}}=T_{(N-S-1)}$. In contrast, the LRC completion time, $T_{\mathrm{LRC}}$, is the earliest moment the received responses—augmented by applicable local repairs—satisfy the $K_c$ demand. Because the LRC naturally retains the baseline global recovery threshold, $T_{\mathrm{LRC}} \le T_{\mathrm{RS}}$. Strict inequality ($T_{\mathrm{LRC}} < T_{\mathrm{RS}}$) occurs whenever favorable straggler patterns allow local repair to finish the computation early, reducing latency despite both codes imposing the identical per-worker computation load $M$.

For the numerical evaluation, we consider $N=K=12$, $S=5$,
$K_c=3$, and $M=S+K_c=8$. The worker computation times are
generated independently according to \eqref{eq:worker-latency-model},
with $\tau=0.4~\mathrm{s}$ and $\lambda=0.4~\mathrm{s}^{-1}$.
Fig.~\ref{fig:rs-lrc-latency} compares the empirical cumulative
distribution functions (CDFs) of the completion times
$T_{\mathrm{RS}}$ and $T_{\mathrm{LRC}}$, estimated over
$10^5$ independent realizations. The leftward shift of the
LRC-based CDF indicates a higher probability of completing
the computation earlier. The average completion time decreases
from approximately $2.45~\mathrm{s}$ for the RS-based construction to
$2.17~\mathrm{s}$ for the LRC-based construction, corresponding to an
average latency reduction of approximately $11.43\%$.
Moreover, the LRC approach finishes strictly earlier than the RS approach in $67.67\%$ of the realizations, clearly demonstrating the benefit of pattern-dependent local repair.

\begin{figure}[t]
    \centering
    \includegraphics[width=0.7\columnwidth]{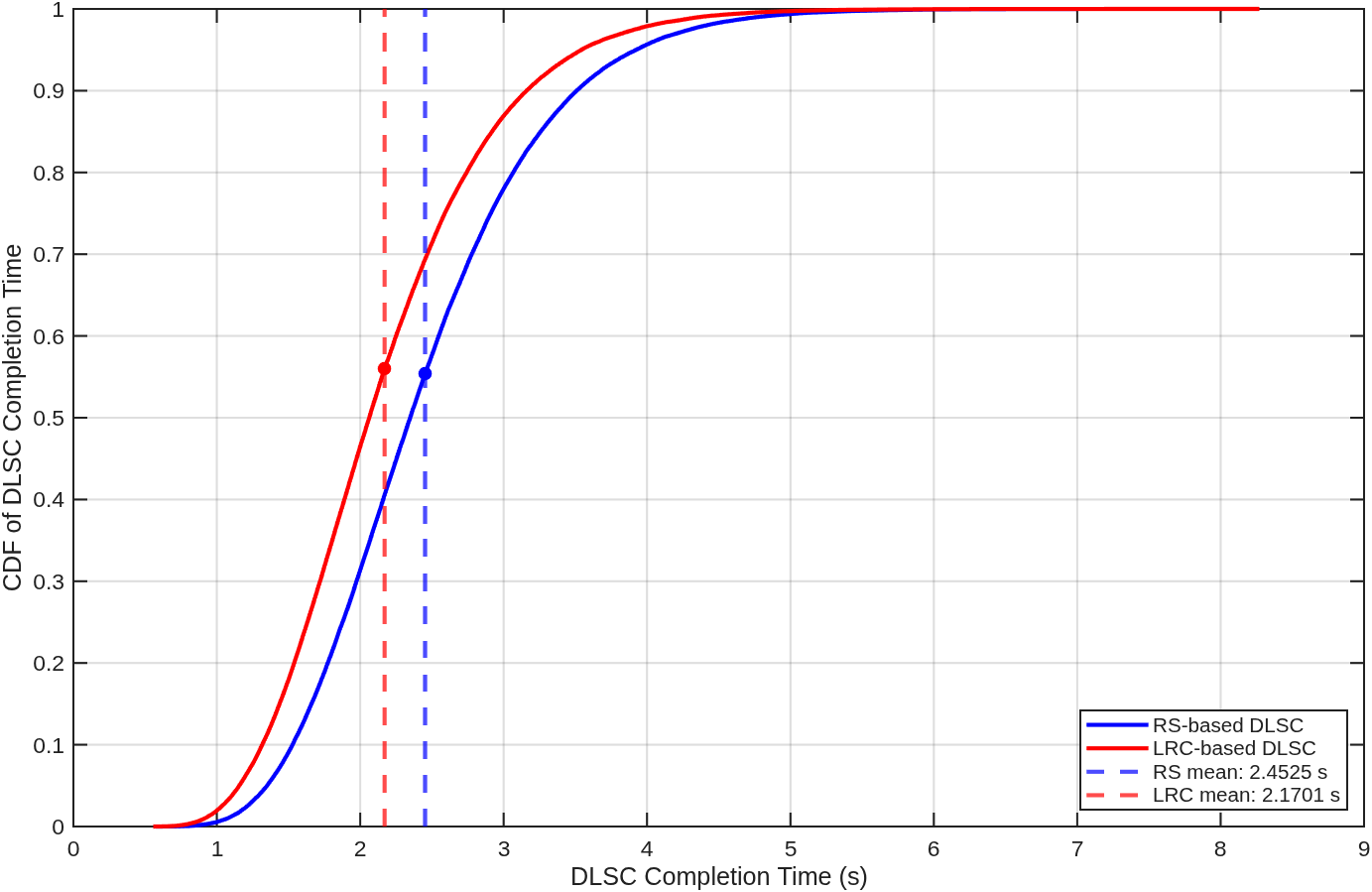}
    \caption{\small Empirical completion-time distributions of the RS- and LRC-based cyclic DLSC constructions under shifted-exponential worker computation times. The LRC approach achieves strictly earlier completion for favorable worker-response patterns by exploiting local repair.}
    \label{fig:rs-lrc-latency}
\end{figure}
\endgroup
\section{Communication-Efficient Cyclic DLSC Construction}
\label{sec:comm-reduced}
\begingroup
\setlength{\abovedisplayskip}{4pt}
\setlength{\belowdisplayskip}{4pt}
\setlength{\abovedisplayshortskip}{3pt}
\setlength{\belowdisplayshortskip}{3pt}

We next extend the proposed construction to reduce the per-worker communication load by a factor of $m$. Assume that $m\!\mid\!L$ and $mK_c\leq N-S$. For ease of presentation, let $K=N$. Each message $P_k\in\mathbb F^{L\times 1}$ is partitioned into $m$ equal-length \emph{submessages},
\begin{equation*}
\begingroup
\setlength{\arraycolsep}{1pt}
\renewcommand{\arraystretch}{0.75}
P_k^T=
\begin{bmatrix}
P_{k,1}^T & P_{k,2}^T & \cdots & P_{k,m}^T
\end{bmatrix},
\quad P_{k,t}\in\mathbb F^{{L/m}\times 1}.
\endgroup
\end{equation*}
Each worker transmits one coded response of length $L/m$.

To accommodate the $m$ submessages of each of the $K_c$ demanded
linear combinations, we employ $mK_c$ nested cyclic codes $\mathcal C_1\subsetneq\mathcal C_2\subsetneq\cdots
\subsetneq\mathcal C_{mK_c},$
with minimum distances and dimensions 
\[
d_i=S+mK_c-i+1,\qquad
k_i=k_1+i-1,\qquad i \in [1:mK_c].
\]

Similar to the general DLSC construction in Sec.~\ref{subsubsec:general-nested-construction}, we select minimum-distance codewords $\mathbf c_1,\ldots,\mathbf c_{mK_c}$ with nested supports. As before, the data assignment is specified by $\operatorname{supp}(\mathbf c_1)$ and its successive cyclic shifts. Since $\operatorname{wt}(\mathbf c_1)=d_1=S+mK_c$, each data partition is assigned to exactly $S+mK_c$ workers. The cyclic shifts distribute this data uniformly, yielding a per-worker computation load of $M=S+mK_c$.

Let $\mathbf G$ be a generator matrix of the supercode $\mathcal C_{mK_c}$. Following the general DLSC framework, we form the block matrix $\mathbf\Delta\in\mathbb F^{N\times NmK_c}$ using the $mK_c$ selected minimum-weight codewords and their cyclic shifts. We then factor $\mathbf\Delta=\mathbf G^T\mathbf B',$ where $\mathbf B'= [\mathbf B'_0\;\mathbf B'_1\;\cdots\;\mathbf B'_{N-1}]$, such that the bottom $mK_c\times mK_c$ submatrix of each block $\mathbf B'_j$ is nonsingular.

For the demand matrix $\mathbf A\in\mathbb F^{K_c\times N}$, consider the Kronecker product $\mathbf A\otimes\mathbf I_m \in\mathbb F^{mK_c\times Nm}$. Exploiting the nonsingularity of the corresponding submatrices of $\mathbf B'_j$, each block of $m$ consecutive columns $[jm:(j+1)m-1]$ from $\mathbf A\otimes\mathbf I_m$ can be embedded into the bottom $mK_c$ rows of a suitable matrix $\mathbf B_j$ whose columns lie in $\operatorname{Col}(\mathbf B'_j)$. Stacking these matrices yields $\mathbf B= [\mathbf B_0\;\mathbf B_1\;\cdots\;\mathbf B_{N-1}],$ where the bottom rows of $\mathbf B_j\in\mathbb F^{k_{mK_c}\times m}$ satisfy:
\begin{equation}
\mathbf B_j
(k_{mK_c}-mK_c+1:k_{mK_c},:)
=
\mathbf a_j\otimes\mathbf I_m,
\label{eq:comm-demand-embedding}
\end{equation}
with $\mathbf a_j$ denoting column $j$ of $\mathbf A$. The resulting encoding matrix is $\mathbf\Phi=\mathbf G^T\mathbf B$. Because the $m$ columns in $\mathbf B_j$ is strictly drawn from the column space of the corresponding $\mathbf B'_j$, the prescribed cyclic zero pattern—and thus the cyclic data assignment—is perfectly preserved.

The responses transmitted by the workers are given by 
\[
\mathbf Y = \mathbf\Phi [P_{0,1}\ \cdots\ P_{0,m}\ \cdots\ P_{N-1,1}\ \cdots\ P_{N-1,m}]^T.
\]
As in the baseline construction, the master can recover 
\[
\mathbf{B} [P_{0,1}\ \cdots\ P_{0,m}\ \cdots\ P_{N-1,1}\ \cdots\ P_{N-1,m}]^T
\]
by collecting responses from any $N-S$ workers. By focusing on the bottom $mK_c$ rows, the master obtains
\[
(\mathbf A\otimes\mathbf I_m)[P_{0,1}\ \cdots\ P_{0,m}\ \cdots\ P_{N-1,1}\ \cdots\ P_{N-1,m}]^T.
\]
This resulting matrix contains exactly the $m$ submessages for each of the $K_c$ desired linear combinations. By concatenating these submessages, $\mathbf A[P_0\ P_1\ \cdots\ P_{N-1}]^{T}$ can be reconstructed.

\begin{theorem}
\label{thm:comm-reduced-cyclic-dlsc}
Let $mK_c\leq N-S$, and consider a family of
$mK_c$ nested length-$N$ cyclic codes
$\mathcal{C}_1\subsetneq\cdots\subsetneq
\mathcal{C}_{mK_c}$ with minimum distances
$d_i=S+mK_c-i+1$, $i\in[1:mK_c]$,
whose selected minimum-distance codewords
have nested supports. Then, for $K=N$ and
any full-row-rank demand matrix
$\mathbf{A}\in\mathbb{F}^{K_c\times N}$,
the proposed communication-efficient construction
achieves a per-worker computation load
$M=S+mK_c$ and communication load $L/m$,
while guaranteeing deterministic recovery
from any $N-S$ workers.
\end{theorem}

\begin{remark}\normalfont
For $K_c=1$, the proposed communication-efficient construction generalizes the gradient coding scheme in~\cite{Ye2018}. In~\cite{Ye2018}, polynomials with prescribed roots and their lifted counterparts are evaluated at distinct field elements to generate the worker encoding coefficients. From a coding-theoretic perspective, this amounts to using a Vandermonde generator matrix to produce codewords whose prescribed zero positions enforce the data assignment. Our framework bypasses this specific Vandermonde requirement. By utilizing nested cyclic codes and obtaining the required data-assignment patterns from their minimum-distance codewords, we enable the use of more general generator matrices. When restricted to RS codes for $K_c=1$, our framework reduces exactly to the Vandermonde-based approach of~\cite{Ye2018}, sharing its identical communication--computation tradeoff $M=S+m$ and hard straggler threshold of $N-S$.
\end{remark}

For $K=\mu N$, where $\mu$ is a positive integer, the construction naturally extends by partitioning the $K$ messages into $\mu$ disjoint subsets of size $N$ and applying the cyclic assignment independently to each. Because the desired computation is linear, each worker simply sums its locally encoded responses across all $\mu$ subsets prior to transmission. This scales the total per-worker computation load to $M=\mu(S+mK_c)=\frac{K}{N}(S+mK_c)$, while strictly maintaining the reduced communication load of $L/m$.

Furthermore, for $N-S\geq mK_c$ and sufficiently large $K$, we establish a matching converse (proved in Appendix~\ref{app:comm-converse}). Any cyclic DLSC scheme with communication load $L/m$ that guarantees the recovery of any $K_c$ linear combinations from any $N-S$ workers must satisfy $M\geq\frac{K}{N}(S+mK_c)$. This formally establishes the aymptotic computation-load optimality of the proposed construction.
\endgroup

\section{Conclusion}
In this work, we developed a DLSC framework based on nested cyclic codes. Our approach maintains a fixed, cyclic data assignment while guaranteeing the deterministic recovery of arbitrary prescribed linear combinations. The RS- and LRC-based constructions provide explicit realizations of this framework, with the latter exploiting local repair to enable recovery beyond the worst-case straggler threshold. Furthermore, we extended the construction to offer a flexible tradeoff between computation and communication loads. Finally, by deriving matching asymptotic converse bounds, we formally established the computation-load optimality of the proposed framework.

\subsection*{AI Statement}
The authors used OpenAI's ChatGPT (GPT-6) and Google's Gemini 3.1 Pro during the preparation
of this manuscript for grammar correction, sentence
rephrasing, and the generation of LaTeX code
for formatting the matrices $\mathbf{B}_j'$ (in Appendix~\ref{app:Bj-nonsingular}) and
$\mathbf G_{\rm LRC}$ (in Section~\ref{sec:cyclic-dlsc-framework}).
The authors reviewed, verified, and edited the
final content and take full responsibility for
the accuracy and integrity of the manuscript.

\bibliographystyle{ieeetr}
\bibliography{ref}
\appendices

\section{Nonsingularity of the Bottom $K_c\times K_c$
Submatrix of $\mathbf B'_j$}
\label{app:Bj-nonsingular}

For each $j\in[0:N-1]$ and $i\in[1:K_c]$,
let
$\mathbf b_i^{(j)}
=[b^{(j)}_{1,i},\ldots,b^{(j)}_{k_i,i}]^T$
denote the message vector corresponding to
the $i$-th codeword in $\mathbf\Delta_j$.

We first show that for every cyclically shifted version of the selected minimum-distance codeword $\mathbf{c}_1 \in \mathcal{C}_1$, the last coefficient of its corresponding message vector remains nonzero. Specifically, our goal is to establish that $b^{(j)}_{k_1,1} \neq 0$ for all $j \in [0:N-1]$. Recall that the top $k_1$ rows of $\mathbf{G}$ are ordered such that for the unshifted codeword $\mathbf{c}_1$, the message vector satisfies $b^{(0)}_{k_1,1} \neq 0$.

Let $\mathbf{g}_e = [1, \alpha^{s_e}, \ldots, \alpha^{(N-1)s_e}]$ denote the row indexed by $e \in [0:k_{K_c}-1]$ in $\mathbf{G}$, where $s_e$ is the corresponding exponent and $\alpha$ is a primitive $N$-th root of unity. Since
\[
\mathbf{c}_1 = \sum_{e=0}^{k_1-1} b^{(0)}_{e+1,1}\mathbf{g}_e,
\]
its $j$-th left cyclic shift has message coefficients
\[
b^{(j)}_{e+1,1} = \alpha^{j s_e}b^{(0)}_{e+1,1}.
\]
Consequently,
\[
b^{(j)}_{k_1,1} = \alpha^{j s_{k_1-1}}b^{(0)}_{k_1,1} \neq 0, \quad j\in[0:N-1],
\]
since every power of $\alpha$ is nonzero.



\begin{lem}
\label{lem:Bj-nonsingular}
Let $\mathbf{G}\in\mathbb{F}^{k_{K_c}\times N}$ be a generator matrix for the supercode $\mathcal{C}_{K_c}$, with rows $\mathbf{g}_e = [1\ \alpha^{s_e}\ \cdots\ \alpha^{(N-1)s_e}]$ for $e \in [0:k_{K_c}-1]$, where $\alpha$ is a primitive $N$-th root of unity such that the first
$k_i$ rows span $\mathcal C_i$ for each
$i\in[1:K_c]$.
Pad each message vector $\mathbf b_i^{(j)}$
with $k_{K_c}-k_i$ trailing zeros and
arrange the resulting vectors as the
columns of $\mathbf B'_j$. The resulting
matrix has the form
\[
\begingroup
\setlength{\arraycolsep}{2pt}
\renewcommand{\arraystretch}{0.82}
\mathbf{B}_j'\triangleq
\left[
\begin{array}{cccc}
b^{(j)}_{1,1} & b^{(j)}_{1,2} & \cdots
& b^{(j)}_{1,K_c}\\
b^{(j)}_{2,1} & b^{(j)}_{2,2} & \cdots
& b^{(j)}_{2,K_c}\\
\vdots & \vdots & \ddots & \vdots\\
b^{(j)}_{k_{K_c}-K_c,1}
&
b^{(j)}_{k_{K_c}-K_c,2}
&
\cdots
&
b^{(j)}_{k_{K_c}-K_c,K_c}\\
\hline
b^{(j)}_{k_{K_c}-K_c+1,1}
&
b^{(j)}_{k_{K_c}-K_c+1,2}
&
\cdots
&
b^{(j)}_{k_{K_c}-K_c+1,K_c}\\
0
&
b^{(j)}_{k_{K_c}-K_c+2,2}
&
\cdots
&
b^{(j)}_{k_{K_c}-K_c+2,K_c}\\
\vdots
&
\ddots
&
\ddots
&
\vdots\\
0
&
\cdots
&
0
&
b^{(j)}_{k_{K_c},K_c}
\end{array}
\right].
\endgroup
\]
The submatrix formed by the last
$K_c$ rows of $\mathbf B'_j$ is
nonsingular for every $j\in[0:N-1]$.
\end{lem}

\begin{proof}
Since $k_i=k_{K_c}-K_c+i$, the bottom
$K_c\times K_c$ submatrix of $\mathbf B'_j$
is upper triangular, with diagonal entries
$b^{(j)}_{k_i,i}$, for $i\in[1:K_c]$. For $i=1$, we have
$b^{(j)}_{k_1,1}\neq0$ by the preceding
cyclic-shift argument.

For $i\in[2:K_c]$, the selected codeword
has weight $d_i$. Since $d_{i}=d_{i-1}-1$,
it cannot belong to $\mathcal C_{i-1}$. As the first
$k_{i-1}=k_i-1$ generator matrix rows span
$\mathcal C_{i-1}$, the coefficient
corresponding to the $k_i$-th generator matrix
row must be nonzero. Hence,
\[
b^{(j)}_{k_i,i}
=b^{(j)}_{k_{K_c}-K_c+i,i}
\neq0,\quad i \in [2:K_c].
\]
Thus, the bottom $K_c\times K_c$ submatrix
of $\mathbf B'_j$ is upper triangular
with nonzero diagonal entries and is
therefore nonsingular.
\end{proof}

\section{Proof of Theorem~\ref{thm:converse}}
\label{app:converse}
\begingroup
\setlength{\abovedisplayskip}{4pt}
\setlength{\belowdisplayskip}{4pt}
\setlength{\abovedisplayshortskip}{3pt}
\setlength{\belowdisplayshortskip}{3pt}

We proceed by contradiction. Suppose that
\begin{equation}
M<\frac{K}{N}(S+K_c).
\label{eq:contra-assumption}
\end{equation}

Under the cyclic data-assignment structure, every data fragment is
assigned to the same number of workers. Since each of the $N$ workers
is assigned $M$ data fragments, the total number of worker-data
assignments is $NM$. Therefore, each of the $K$ data fragments is
assigned to $\frac{NM}{K}$ workers. From \eqref{eq:contra-assumption}, $\frac{NM}{K}<S+K_c.$ So every data fragment is assigned to at most $S+K_c-1$
workers. Consequently, every data fragment is not assigned to at least $\rho\triangleq N-S-K_c+1$
workers. Moreover, $\rho\geq1$ because $N-S \geq K_c$.

For each data fragment $D_j$, choose any $\rho$ workers to which it is
not assigned, and denote their indices by $\mathcal{Z}_j\subseteq[0:N-1]$.
There are exactly $\binom{N}{\rho}$ possible choices for
$\mathcal{Z}_j$. By the hypothesis $K>(K_c-1)\binom{N}{\rho},$
the pigeonhole principle guarantees the existence of at least $K_c$
data fragments sharing the same selected set of $\rho$ unassigned workers.

Let $\mathcal{J}\subseteq[0:K-1]$, with $|\mathcal{J}|=K_c$, denote the
indices of any such $K_c$ data fragments, and let
$\mathcal{Z}\subseteq[0:N-1]$, with $|\mathcal{Z}|=\rho$, denote their
common set of unassigned workers. Since workers in $\mathcal{Z}$ do not
store any of the data fragments indexed by $\mathcal{J}$, the
corresponding entries of the encoding matrix must be zero. Thus,
\begin{equation}
\mathbf{\Phi}(\mathcal{Z},\mathcal{J})=\mathbf{0}.
\label{eq:zero-submatrix}
\end{equation}

Since $N-\rho=S+K_c-1\geq S$, we can choose a straggler set $\mathcal{F}\subseteq[0:N-1]\setminus\mathcal{Z}$ with $|\mathcal{F}|=S$. Thus, $\mathcal{R}=[0:N-1]\setminus\mathcal{F}$ denotes the indices of the non-straggling workers. Then, $|\mathcal{R}|=N-S$ and $\mathcal{Z}\subseteq\mathcal{R}$.

Consider the submatrix $\mathbf{\Phi}({\mathcal{R},\mathcal{J}})$. By \eqref{eq:zero-submatrix}, its $\rho$ rows indexed by $\mathcal{Z}$ are identically zero. Therefore, the number of potentially nonzero rows is at most 
\begin{equation*}
|\mathcal{R}|-\rho = N-S-(N-S-K_c+1) = K_c-1.
\end{equation*}
Consequently,
\begin{equation*}
\operatorname{rank}\bigl(\mathbf{\Phi}({\mathcal{R},\mathcal{J}})\bigr) \leq K_c-1.
\label{eq:rank-deficiency}
\end{equation*}

Choose a demand matrix
$\mathbf A\in\mathbb F^{K_c\times K}$
such that $\operatorname{rank}(\mathbf{A(:,\mathcal{J}}))=K_c.$ If the desired linear combinations were recoverable from the workers
in $\mathcal{R}$, then $\operatorname{Row}(\mathbf A)
\subseteq
\operatorname{Row}(\mathbf\Phi(\mathcal R,:)).$
Restricting both row spaces to the coordinates indexed by
$\mathcal J$ would therefore imply $\operatorname{Row}(\mathbf A(:,\mathcal J))
\subseteq
\operatorname{Row}(\mathbf\Phi(\mathcal R,\mathcal J))$. However, this is impossible since $\operatorname{rank}(\mathbf{A}(:,\mathcal{J})) = K_c$ while $\operatorname{rank}(\mathbf{\Phi}(\mathcal{R},\mathcal{J})) \leq K_c-1$. Hence, the assumption in \eqref{eq:contra-assumption} is false, and
therefore:
\begin{equation*}
M\geq\frac{K}{N}(S+K_c).
\label{eq:cyclic-converse-bound}
\end{equation*}
\endgroup

\section{Converse for the Communication-Efficient Setting}
\label{app:comm-converse}
\begingroup
\setlength{\abovedisplayskip}{4pt}
\setlength{\belowdisplayskip}{4pt}
\setlength{\abovedisplayshortskip}{3pt}
\setlength{\belowdisplayshortskip}{3pt}

We proceed by contradiction. Suppose that
\begin{equation*}
M<\frac{K}{N}(S+mK_c).
\end{equation*}
Following the counting and pigeonhole arguments in
Appendix~\ref{app:converse}, with $K_c$ replaced by
$mK_c$, each data fragment
is absent from at least $\rho=N-S-mK_c+1$
workers. Since
$K>(K_c-1)\binom{N}{\rho}$,
there exist a set $\mathcal{J}\subseteq[0:K-1]$ with $|\mathcal{J}|=K_c$ and $\mathcal{Z}\subseteq[0:N-1]$ with
$|\mathcal{Z}|=\rho$ such that none of the workers indexed by
$\mathcal{Z}$ stores any data fragment indexed by $\mathcal{J}$.

Choose a straggler set
$\mathcal{F}\subseteq[0:N-1]\setminus\mathcal{Z}$
with $|\mathcal{F}|=S$, and let
$\mathcal{R}=[0:N-1]\setminus\mathcal{F}$.
Then $\mathcal{Z}\subseteq\mathcal{R}$, and at most
\begin{equation}
\label{eq:comm-number-useful-workers}
|\mathcal{R}|-|\mathcal{Z}|
=N-S-\rho=mK_c-1
\end{equation}
non-straggling workers have access to the data fragments
indexed by $\mathcal{J}$.

Choose a demand matrix $\mathbf{A}\in\mathbb{F}^{K_c\times K}$ such that the submatrix $\mathbf{A}(:,\mathcal{J}) = \mathbf{I}_{K_c}$ and all remaining entries are zero. Satisfying this demand requires the master to perfectly recover the set of data fragments $\{P_l\mid l\in\mathcal{J}\}$, which comprises $K_cL$ independent symbols. Among the non-straggling workers in $\mathcal{R}$, those in $\mathcal{Z}$ do not host any of the data fragments from $\mathcal{J}$; therefore, their responses contain no information about these fragments. From \eqref{eq:comm-number-useful-workers}, at most $mK_c-1$ workers can provide relevant information to the master. Since each worker transmits exactly $L/m$ symbols, the master receives at most 
\begin{equation*}
(mK_c-1)\frac{L}{m} = K_cL - \frac{L}{m} < K_cL
\end{equation*}
relevant symbols. This renders the recovery of the $K_cL$ independent symbols information-theoretically impossible. Thus, we conclude that
\begin{equation*}
M\geq\frac{K}{N}(S+mK_c).
\end{equation*}
\endgroup

\section{Finding a Minimum-Weight Codeword in Optimal LRCs}
\label{app:lrc-minweight}
\begingroup
\setlength{\abovedisplayskip}{4pt}
\setlength{\belowdisplayskip}{4pt}
\setlength{\abovedisplayshortskip}{3pt}
\setlength{\belowdisplayshortskip}{3pt}

Consider an $[n,k,d_{\min}]$ code that possesses disjoint $(r,\delta)$ locality with generator matrix $\mathbf{G}\in\mathbb{F}^{k\times n}$. Assume the code is minimum-distance-optimal with respect to the bound in \cite[Theorem~2.1]{Kamath2014}, i.e., its minimum distance satisfies 
$$d_{\min} = n - k +1 - \left(\left\lceil \frac{k}{r} \right\rceil - 1\right)(\delta - 1).$$ 
Write
\[
k=r\left(\left\lceil\frac{k}{r}\right\rceil-1\right)+x,
\qquad 1\leq x\leq r.
\]
The following lemma establishes the linear independence property required to isolate a minimum-weight codeword.

\begin{lem}
\label{lem:lrc-independence}
Choose $\lceil k/r \rceil - 1$ distinct local groups and select $r$ columns of $\mathbf{G}$ from each. Then, select $x$ columns from one additional local group. The resulting $k$ columns are linearly independent.
\end{lem}

\begin{proof}
Let $\mathbf{G}_1$ be the $k\times k$ submatrix formed by these selected columns, and suppose $\operatorname{rank}(\mathbf{G}_1)<k$. This implies there exists a non-zero message vector $\mathbf{m} \in \mathbb{F}^{1 \times k}$ such that $\mathbf{m}\mathbf{G}_1 = \mathbf{0}$. Consequently, the corresponding codeword $\mathbf{c} = \mathbf{m}\mathbf{G}$ has zeros in all $k$ of these selected coordinates. 

Consider the $\lceil k/r \rceil - 1$ local groups where $r$ coordinates were selected. Since each group has size $r+\delta-1$, the presence of $r$ zeros leaves at most $\delta-1$ potentially non-zero coordinates. Because any non-zero restriction of a codeword to a local group must have a weight of at least $\delta$, the restriction to each of these groups must be entirely zero. Thus, $\mathbf{c}$ contains at least
\[
\left(\left\lceil\frac{k}{r}\right\rceil-1\right)
(r+\delta-1)+x
\]
zero coordinates. The weight of $\mathbf{c}$ is therefore at most
\[
n-k-
\left(\left\lceil\frac{k}{r}\right\rceil-1\right)
(\delta-1)
=d_{\min}-1,
\]
which contradicts the minimum-distance optimality of the LRC. Thus, $\mathbf{G}_1$ must be nonsingular.
\end{proof}

Recall that the $k$ columns of $\mathbf{G}_1$ consist of $r$ columns each from $\lceil k/r\rceil-1$ local groups, and $x$ columns from one additional local group. Let the $l$-th column of $\mathbf{G}_1$ correspond to one of the $x$ columns drawn from this additional group. Define $\mathbf{u}\in\mathbb{F}^{1\times k}$ as the unit vector with a $1$ at the $l$-th position and zeros elsewhere, and set the message vector to $\mathbf{m}=\mathbf{u}\mathbf{G}_1^{-1}$. Since $\mathbf{u}\neq\mathbf{0}$, the message vector $\mathbf{m}$ is non-zero and thus the resulting $\mathbf{c}=\mathbf{m}\mathbf{G}$ is non-zero as well. By construction, restricting $\mathbf{c}$ to the $k$ selected coordinates yields $\mathbf{m}\mathbf{G}_1 = \mathbf{u}$.

Consequently, $\mathbf{c}$ has zeros at the $r$ selected positions in each of the  $\lceil k/r\rceil-1$ local groups which are part of $\mathbf{G}_1$, and $x-1$ zeros among the selected positions of the additional group. As the presence of $r$ zeros in a local group forces all its remaining symbols to zero, $\mathbf{c}$ contains at least
\[
\left(\left\lceil\frac{k}{r}\right\rceil-1\right)(r+\delta-1)+x-1
\]
zero coordinates. Recalling that $k=r(\lceil k/r\rceil-1)+x$, the weight of this non-zero codeword is at most
\[
n-\left[\left(\left\lceil\frac{k}{r}\right\rceil-1\right)(r+\delta-1)+x\right]+1 = d_{\min}.
\]
Since the code has minimum distance $d_{\min}$, it must hold that $\operatorname{wt}(\mathbf{c})=d_{\min}$. Thus, we conclude that $\mathbf{c}$ is a minimum-weight codeword.

\end{document}